\documentclass{article}
\usepackage{latexsym,amssymb,amsmath}
\usepackage[pdftex]{graphicx}
\usepackage{pifont}
\newtheorem{lemma}{Lemma}
\newtheorem{proposition}{Proposition}
\newtheorem{definition}{Definition}
\newtheorem{theorem}{Theorem}
\newtheorem{corollary}{Corollary}
\newcommand{\acro}[1]{\ensuremath{\mathcal{#1}}}
\newcommand{\agents}{N}

\newcommand{\img}{\mathbf{image}}
\newcommand{\themap}{\mathbf{f}}
\newcommand{\tuple}[1]{\langle #1\rangle}
\newcommand{\coal}[1]{[#1]}

\newenvironment{proof}[1][Proof]{\begin{trivlist}
\item[\hskip \labelsep {\bfseries #1}]}{\unskip\nobreak\hfil\penalty50
   \hskip2em\hbox{}\nobreak\hfil
   \ding{111}%
   \parfillskip=0pt \finalhyphendemerits=0
    \medskip\goodbreak\noindent\ignorespacesafterend\end{trivlist}}
\begin{document}
\title{Completeness of Epistemic Coalition Logic with Group Knowledge} 
\author{Thomas {\AA}gotnes and Natasha Alechina}
\date{\today}

\maketitle

\begin{abstract}
  Coalition logic is one of the most popular logics for multi-agent
  systems. While epistemic extensions of coalition logic have received
  much attention, existence of their complete axiomatisations
  has so far been an open problem. In this paper we settle several of
  those problems. We prove completeness for epistemic coalition logic
  with common knowledge, with distributed knowledge, and with both
  common and distributed knowledge, respectively.
\end{abstract}

\section{Introduction}

Logics of coalitional ability such as \emph{Coalition Logic}
(\acro{CL}) \cite{pauly:2002a}, \emph{Alternating-time Temporal Logic}
(\acro{ATL}) \cite{Alur2002}, and STiT logics \cite{Belnap88stit}, are
arguably some of the most studied logics in multi-agent systems in
recent years. Many different variants of these logics have been
proposed and studied, but so far meta-logical results have focused
more on computational expressiveness and expressive power and less on
completeness, with Goranko's and van Drimmelen's completeness proof
for \acro{ATL} \cite{goranko:2006a}, Pauly's completeness proof for
\acro{CL} \cite{pauly:2002a} and Broersen and colleagues' completeness
proofs for different variants of STiT logic
\cite{Broersen//:07a,BroersenDALT2008,BroersenHerzigTroquardJANCL2009}
being notable exceptions.

The main construction in coalitional ability logics is of the form
$[G]\phi$, where $G$ is a set of agents and $\phi$ a formula,
intuitively meaning that $G$ is \emph{effective} for $\phi$, or that
$G$ can make $\phi$ come true no matter what the other agents do.  One
of the most studied extension of basic coalitional ability logics is
adding \emph{knowledge} operators of the type found in \emph{epistemic
  logic} \cite{Fagin1995,Meyer1995}: both \emph{individual} knowledge
operators $K_i$ where $i$ is an agent, and different types of
\emph{group} knowledge operators $E_G$, $C_G$ and $D_G$ where $G$ is a
group of agents, standing for everybody-knows, common knowledge and
distributed knowledge, respectively. Combining coalitional ability
operators and epistemic operators in general and group knowledge
operators in particular lets us express many potentially interesting
properties of multi-agent systems, such as \cite{Hoek2003a}:
\begin{itemize}
\item $K_i \phi \rightarrow [\{i\}]K_j \phi$: $i$ can communicate her
  knowledge of $\phi$ to $j$
\item $C_G \phi \rightarrow [G]\psi$: common knowledge in $G$ of
  $\phi$ is sufficient for $G$ to ensure that $\psi$
\item $[G]\psi \rightarrow D_G \phi$: distributed knowledge in $G$ of
  $\phi$ is necessary for $G$ to ensure that $\psi$
\item $D_G\phi \rightarrow [G]E_G\phi$: $G$ can cooperate to make
  distributed knowledge explicit
\end{itemize}

In this paper we study a complete axiomatisation of variants of
\emph{epistemic coalition logic} (\acro{ECL}), extensions of coalition
logic with individual knowledge and different combinations of common
knowledge and distributed knowledge. Coalition logic, the next-time
fragment of \acro{ATL}, is one of the most studied coalitional ability
logics, and this paper settles a key open problem: completeness of its
epistemic variants.

While epistemic coalitional ability logics have been studied to a
great extent, we are not aware of any published completeness results
for such logics with all epistemic operators. \cite{Hoek2003a} gives
some axioms of \acro{ATEL}, \acro{ATL} extended with epistemic
operators, but does not attempt to prove completeness\footnote{In an
  abstract of a talk given at the LOFT workshop in 2004
  \cite{goranko:2004}, the authors propose a full axiomatisation of
  \acro{ATEL} with individual knowledge and common knowledge operators
  with similar axioms. However, neither a completeness proof nor the result
  itself was published, and a proof indeed does not exist as explained
  to us in personal communication (Valentin Goranko).}.  Broersen and
colleagues \cite{BroersenDALT2008,BroersenHerzigTroquardJANCL2009}
prove completeness of variants of STiT logic that include individual
knowledge operators, but not group knowledge operators, and
\cite{BroersenHerzigTroquardJANCL2009} concludes that adding group
operators is an important challenge.

The rest of the paper is organised as follows. In the next section we
first give a brief review of coalition logic, and how it is extended
with epistemic operators. We then, in each of the three following
sections, consider basic epistemic coalition logic with individual
knowledge operators extended with common knowledge, with distributed
knowledge, and with both common and distributed knowledge,
respectively. For each of these cases we show a completeness
result. For the common knowledge case we also show a filtration
result. We conclude in Section \ref{sec:conclusions}.

\section{Background}

We define several extensions of propositional logic, and the usual
derived connectives, such as $\phi \rightarrow \psi$ for $\neg \phi
\wedge \psi$, will be used.

\subsection{Coalition Logic}

Assume a set $\Theta$ of atomic propositions, and a finite set
$\agents$ of agents. A \emph{coalition} is a set $G \subseteq \agents$
of agents. We sometimes abuse notation and write a singleton coalition
$\{i\}$ as $i$.

The language of coalition logic (\acro{CL}) is defined by the
following grammar:
\[\phi ::= p \mid \neg \phi \mid \phi \wedge \phi \mid [G]\phi\]
where $p \in \Theta$ and $G \subseteq \agents$.

A \emph{coalition model} is a tuple
\[M = \tuple{S,E,V}\]
where 
\begin{itemize}
\item $S$ is a non-empty set of \emph{states};
\item $V$ is a \emph{valuation function}, assigning a set $V(s) \subseteq \Theta$ to each state $s\in S$;
\item $E$ assigns a \emph{truly playable effectivity function} $E(s)$ to each state
$s \in S$.
\end{itemize}
An \emph{effectivity function} \cite{pauly:2002a} over $\agents$ and a
set of states $S$ is a function $E$ that maps any coalition $G
\subseteq \agents$ to a set of sets of states $E(G) \subseteq 2^S$. An
effectivity function is \emph{truly playable}
\cite{pauly:2002a,goranko:2011} iff it satisfies the following
conditions:
\begin{description}
\item[E1] $\forall s \in S \forall G \subseteq N \ \emptyset \not \in E(G)(s)$ (Liveness)
\item[E2] $\forall s \in S \forall G \subseteq N \  S \in E(G)(s)$ (Safety)
\item[E3] $\forall s \in S \forall X \subseteq S \ \bar{X} \not \in E(\emptyset)(s) \Rightarrow X \in E(N)(s)$ ($N$-maximality)
\item[E4] $\forall s \in S \forall G \subseteq N \forall X \subseteq Y \subseteq S \ 
 X \in E(G)(s) \Rightarrow Y \in E(G)(s)$ (outcome monotonicity)
\item[E5]  $\forall s \in S \forall G_1,G_2 \subseteq N \forall X,Y \subseteq S\ X \in E(G_1)(s)$ and $Y \in E(G_2)(s)$
$\Rightarrow X \cap Y \in E(G_1 \cup G_2)(s)$, where $G_1 \cap G_2 = \emptyset$ (superadditivity)
\item[E6] $E^{nc}(\emptyset) \not = \emptyset$, where $E^{nc}(\emptyset)$ is
the \emph{non-monotonic core} of the empty coalition, namely
$$\{X \in E(\emptyset): \neg \exists Y (Y \in E(\emptyset)\ {\rm and}\ Y \subset X)\}$$
\end{description}
An effectivity function that only satisfies E1-E5 is called
\emph{playable}. On finite domains an effectivity function is playable
iff it is truly playable \cite{goranko:2011}, because on finite domains E6
follows from E1-E5.

An \acro{CL} formula is interpreted in a state in a coalition model as follows:
\begin{description}
\item[] $M,s \models p$ iff $p \in V(s)$
\item[] $M,s \models \neg \phi$ iff $M,s \not\models \phi$
\item[] $M,s \models (\phi_1 \wedge \phi_2)$ iff $(M,s \models \phi_1 \mbox{ and } M,s \models \phi_2)$
\item[] $M,s \models [G]\phi$ iff $\phi^{M} \in E(s)(G)$
\end{description}
where $\phi^{M} = \{t \in S: M,t \models \phi\}$.

Figure \ref{fig:cl-ax} shows an axiomatisation $CL$ of coalition logic
which is sound and complete wrt. all coalition models
\cite{pauly:2002a}. The following \emph{monotonicity rule} is
derivable, and will be useful later: $\vdash_{CL} \phi \rightarrow
\psi \Rightarrow \vdash_{CL} \coal{G}\phi \rightarrow \coal{G}\psi$.
\begin{figure}[h]
  \centering
\begin{description}
\item[Prop] Classical propositional logic
\item[G1] $\neg [G] \bot$
\item[G2] $[G] \top$
\item[G3] $\neg [\emptyset] \neg \phi \rightarrow [N] \phi$
\item[G4] $[G] (\phi \wedge \psi) \rightarrow [G] \psi$
\item[G5] $[G_1]\phi \wedge [G_2]\psi \rightarrow [G_1 \cup G_2] (\phi \wedge \psi)$, if $G_1 \cap G_2 = \emptyset$
\item[MP] $\vdash_{CL} \phi, \phi \rightarrow \psi \Rightarrow \vdash_{CL}  \psi$
\item[RG] $\vdash_{CL} \phi \leftrightarrow \psi \Rightarrow \vdash_{CL} [G] \phi \leftrightarrow [G] \psi $
\end{description}
  \caption{$CL$: axiomatisation of \acro{CL}.}
  \label{fig:cl-ax}
\end{figure}

\subsection{Adding Knowledge Operators}

Epistemic extensions of coalition logic were first proposed in
\cite{Hoek2003a}\footnote{In that paper for \acro{ATL}; \acro{CL} is a
  fragment of \acro{ATL}.}. They are obtained by extending the
language with \emph{epistemic operators}, and the models with
\emph{epistemic accessibility relations}.

An epistemic accessibility relation for agent $i$ over a set of states
$S$ is a binary relation $\sim_i \subseteq S \times S$. We will assume
that epistemic accessibility relations are equivalence relations. An
\emph{epistemic coalition model}, henceforth often called simply a
\emph{model}, is a tuple
\[M = \tuple{S,E,\sim_1,\ldots,\sim_n,V}\] where $\tuple{S,E,V}$ is a
coalition model and $\sim_i$ is an epistemic accessibility relation
over $S$ for each agent $i$.

Epistemic operators come in two types: individual knowledge operators
$K_i$, where $i$ is an agent, and group knowledge operators $C_G$ and
$D_G$ where $G$ is a coalition for expressing \emph{common knowledge}
and \emph{distributed knowledge}, respectively. Formally, the language
of \acro{CLCD} (\emph{coalition logic with common and distributed
  knowledge}), is defined by extending coalition logic with all of
these operators:
\[\phi ::= p \mid \neg \phi \mid \phi \wedge \phi \mid [H]\phi\mid K_i \phi \mid C_G \phi \mid D_G \phi\]
where $p \in \Theta$, $i \in \agents$, $H \subseteq \agents$ and
$\emptyset \neq G \subseteq \agents$. When $G$ is a coalition, we
write $E_G \phi$ as a shorthand for $\bigwedge_{i \in G} K_i \phi$
(everyone in $G$ knows $\phi$).

The languages of the logics \acro{CLK}, \acro{CLC} and \acro{CLD} are
the restrictions of this language with no $C_G$ and no $D_G$
operators, no $D_G$ operators, and no $C_G$ operators, respectively.

The interpretation of these languages in an (epistemic coalition)
model $M$ is defined by adding the following clauses to the definition
for \acro{CL}:
\begin{description}
\item[] $M,s \models K_i \phi$ iff $\forall t \in S, (s,t) \in \sim_i \Rightarrow M,t \models \phi$
\item[] $M,s \models C_G \phi$ iff $\forall t \in S, (s,t) \in (\bigcup_{i \in
    G} \sim_i)^* \Rightarrow M,t \models \phi$
\item[] $M,s \models D_G \phi$ iff $\forall t \in S, (s,t) \in \bigcap_{i \in G} \sim_i \Rightarrow M,t \models \phi$
\end{description}
where $R^*$ denotes the transitive closure of the relation $R$.  We
use $\models \phi$ to denote the fact that $\phi$ is \emph{valid},
i.e., that $M,s \models \phi$ for all $M$ and states $s$ in $M$.

\subsubsection{Some Auxiliary Definitions}

The following are some auxiliary concepts that will be useful in the
following.

A \emph{pseudomodel} is a tuple $M=(S, \{\sim_i: i \in N\}, \{R_G :
\emptyset \neq G \subseteq \agents\}, E,V)$ where $(S, \{\sim_i: i \in
N\}, E,V)$ is a model and:
\begin{itemize}
\item $R_G \subseteq S \times S$ is an equivalence relation for each $G$
\item For any $i \in \agents$, $R_i = \sim_i$
\item For any $G$, $H$, $G \subseteq H$ implies that $R_H \subseteq R_G$
\end{itemize}
The interpretation of a \acro{CLCD} formula in a state of a
pseudomodel is defined as for a model, except for the case for $D_G$
which is interpreted by the $R_G$ relation:
\begin{description}
\item[] $M,s \models D_G \phi$ iff $\forall t \in S, (s,t) \in R_G \Rightarrow M,t \models \phi$
\end{description}

An \emph{epistemic model} is a model without the $E$ function, i.e., a
tuple $\tuple{S,\sim_1,\ldots,\sim_n,V}$. An \emph{epistemic
  pseudomodel} is a pseudomodel without the $E$ function, i.e., a
tuple $\tuple{S, \{\sim_i: i \in N\}, \{R_G : \emptyset \neq G
  \subseteq \agents\}, V}$.

Finally, a \emph{playable (pseudo)model} is a (pseudo)model where only 
conditions E1-E5 on $E$ hold.

\section{Coalition Logic with Common Knowledge}
\label{sec:clc}

In this section we consider the logic \acro{CLC}, extending coalition
logic with individual knowledge operators and common knowledge. We
first prove a completeness result, and then show that \acro{CLC}
admits filtrations.

\subsection{Completeness}

The axiomatisation $CLC$ is shown in Figure \ref{fig:clc-ax}. It extends
$CL$ with standard axioms and rules for individual and common knowledge
(see, e.g., \cite{Fagin1995}).

\begin{figure}[h]
  \centering
\begin{description}
\item[Prop] Classical propositional logic
\item[G1] $\neg [G] \bot$
\item[G2] $[G] \top$
\item[G3] $\neg [\emptyset] \neg \phi \rightarrow [N] \phi$
\item[G4] $[G] (\phi \wedge \psi) \rightarrow [G] \psi$
\item[G5] $[G_1]\phi \wedge [G_2]\psi \rightarrow [G_1 \cup G_2] (\phi \wedge \psi)$, if $G_1 \cap G_2 = \emptyset$
\item[MP] $\vdash_{CLC} \phi, \phi \rightarrow \psi \Rightarrow \vdash_{CLC}  \psi$
\item[RG] $\vdash_{CLC} \phi \leftrightarrow \psi \Rightarrow \vdash_{CLC} [G] \phi \leftrightarrow [G] \psi $
\item[K] $K_i (\phi \rightarrow \psi) \rightarrow (K_i \phi \rightarrow K_i \psi)$
\item[T] $K_i \phi \rightarrow \phi$
\item[4] $K_i \phi \rightarrow K_i K_i \phi$
\item[5] $\neg K_i \phi \rightarrow K_i \neg K_i \phi$
\item[C1] $E_G \phi \leftrightarrow  \bigwedge_{i \in G} K_i \phi$
\item[C2] $C_G  \phi \rightarrow E_G(\phi \wedge C_G \phi)$
\item[RN] $\vdash_{CLC} \phi \Rightarrow \vdash_{CLC} K_i \phi$
\item[RC] $\vdash_{CLC} \phi \rightarrow E_G(\phi \wedge \psi) \Rightarrow \vdash_{CLC} \phi \rightarrow C_G \psi$
\end{description}
  \caption{$CLC$: axiomatisation of \acro{CLC}.}
  \label{fig:clc-ax}
\end{figure}

It is easy to show that $CLC$ is sound wrt. all models.
\begin{lemma}[Soundness]
  For any $CLC$-formula $\phi$, $\vdash_{CLC} \phi \Rightarrow  \models \phi$.
\end{lemma}
In the remainder of this section we show that $CLC$ also is complete.

\begin{theorem}
\label{th:clc}
Any $CLC$-consistent formula is satisfied in some model.
\end{theorem}
\begin{proof}
  We define a canonical playable model $M^c=(S^c, \{\sim^c_i: i \in N\},
  E^c,V^c)$ as follows:
\begin{description}
\item $S^c$ is the set of all maximally $CLC$ consistent sets of formulas
\item $s \sim^c_i t$ iff $\{\psi: K_i \psi \in s\}=\{\psi: K_i \psi \in t\}$
  \item $X \in E^c(G)(s)$ iff $\left\{\begin{array}{ll}
        \exists \phi \{s \in S^c : \phi \in s\} \subseteq X : [G]\phi \in s& G \neq N\\
        \forall \phi \{s \in S^c : \phi \in s\} \subseteq S^c\setminus X : [\emptyset]\phi \not\in s& G = N
        \end{array}\right.$
\item $V^c$: $s\in V^c(p)$ iff $p \in s$
\end{description}
The conditions on $\sim^c_i$ (that it is an equivalence relation) and on $E^c$ (that it satisfies E1-E5)
hold in $M^c$. The proof for $\sim^c_i$ is obvious and the proof for $E^c$ is standard.
The intuition of cause is that a formula belongs to a state
$s$ in a model iff it is true there (truth lemma).  However, the
canonical model is in general \emph{not} guaranteed to satisfy
every consistent formula in the \acro{CLC} language; the case of $C_G$ in the
truth lemma does not necessarily hold.  Therefore we are going to
transform $M^c$ by filtration into a finite model for a given $CLC$ consistent formula $\phi$.
Note that since $\phi$ is consistent, it will belong to at least one
$s$ in $M^c$.

Let $cl(\phi)$ be the set of subformulas of $\phi$ closed under single
negations and the condition that $C_G \psi \in cl(\phi) \Rightarrow K_i C_G \psi \in
cl(\phi)$ for all $i \in G$.  We are going to filtrate $M^c$ through
$cl(\phi)$. The resulting model $M^f=(S^f, \{\sim^f_i: i \in N\}, E^f,V^f)$ is constructed as follows:
\begin{description}
\item $S^f$ is $\{[s]_{cl(\phi)}: s \in S^c\}$ where $[s]_{cl(\phi)} = 
s \cap cl(\phi)$. We will omit the subscript $cl(\phi)$ in what follows for
readability.
\item $[s] \sim_i^f [t]$ iff $\{\psi: K_i \psi \in [s] \} = \{\psi: K_i \psi \in [t] \}$
\item $V^f(p) = [V^c(p)]_{cl(\phi)}$ (where $[X]_{cl(\phi)}=\{[s] : s \in X\}$).
Again we will omit the subscript for readability.
\item $X \in E^f(G)([s])$ iff $\{s':  \phi_X \in s' \} \in E^c(G)(s)$ 
where $\phi_X = \vee_{[t] \in X} \phi_{[t]}$ and $\phi_{[t]}$ is
a conjunction of all formulas in $[t]$.
\end{description}

We now prove by induction on the size of $\theta$ that for every $\theta \in cl(\phi)$, $M^f, [s] \models \theta$ iff $\psi \in [s]$.

\begin{description}
\item[case $\theta=p$] trivial
\item[case booleans] trivial

\item[case] $\theta=K_i \psi$ 
assume $M^f,[s] \not \models K_i \psi$. The latter means there is a $[s']$
such that $[s] \sim_i^f [s']$ and $M^f,[s'] \not \models \phi$. By the inductive hypothesis
$\phi \not \in [s']$. Since $[s']$ is deductively closed wrt $cl(\phi)$ and $K_i \psi \in cl(\phi)$,
also $K_i \psi \not \in [s']$. $[s] \sim_i^f [s']$ means that $[s]$ and $[s']$ contain the same $K_i$ formulas from $cl(\phi)$, 
hence $K_i \psi \not \in [s]$.

Assume $M^f,[s] \models K_i \psi$. Then for all $[s']$
such that $[s] \sim_i^f [s']$, $M^f,[s'] \models \psi$. This means by the IH that $\psi \in [s']$ for all $[s']\sim_i^f [s]$.
Assume by contradiction that $K_i \psi \not \in [s]$. Then $\phi_{[s]}$, where $\phi_{[s]}$ is the conjunction of all formulas
in $[s]$, is consistent with $\neg K_i \psi$. If we write $\langle K_i \rangle$ for the dual of the $K_i$ modality, this is
equivalent to: $\phi_{[s]} \wedge \langle K_i \rangle \neg \psi$ is consistent. By forcing choices, 
$$\phi_{[s]} \wedge \langle K_i \rangle \bigvee_{\neg \psi \in [t]} \phi_{[t]}$$ is consistent. By the distributivity of $\langle K_i \rangle$ over
$\vee$, 
$$\bigvee_{\neg \psi \in [t]}(\phi_{[s]} \wedge \langle K_i \rangle \phi_{[t]})$$ is consistent. So for some $[t]$ with $\neg \psi \in [t]$,
$\phi_{[s]} \wedge \langle K_i \rangle \phi_{[t]}$ is consistent. We claim that $[s]\sim_i^f [t]$. If this is the case, we have a contradiction,
since we assumed that $\psi \in [s']$ for all $[s']\sim_i^f [s]$.

Proof of the claim: if $\phi_{[s]} \wedge \langle K_i \rangle \phi_{[t]}$ is consistent, then $[s]\sim_i^f [t]$. Suppose not $[s]\sim_i^f [t]$,
that is there is a formula $\chi$ such that $K_i \chi \in [s]$ and $\neg K_i \chi \in [t]$ or vice versa. Then we have $K_i \chi \wedge \phi_{[s]} 
\wedge \langle K_i \rangle (\neg K_i \chi \wedge  \phi_{[t]})$ is consistent, but since $K_i$ is an S5 modality, this is impossible.
Same for the case when $\neg K_i \chi \in [s]$ and $K_i \chi \in [t]$.

\item[case] $\theta=[G]\psi$

  $M^f,[s] \models [G]\psi$ iff $\psi^{M^f} \in E^f(G)([s])$ iff
  $\{s': (\vee_{[t] \in \psi^{M^f}} \phi_{[t]}) \in s'\} \in
  E^c(G)(s)$ iff (by the IH) $\{s': (\vee_{\psi \in [t]} \phi_{[t]})
  \in s'\} \in E^c(G)(s)$ iff(*) $\{s': \psi \in s'\} \in E^c(G)(s)$
  iff(**) $[G]\psi \in s$ iff (since $[G]\psi \in cl(\phi)$) $[G]\psi
  \in [s]$.

  Proof of (*): assume $S^f$ contains $n+k$ states,
  $[t_1],\ldots,[t_n]$ contain $\psi$ and $[s_1],\ldots,[s_k]$
  contain $\neg \psi$. Clearly, $\phi_{[t_1]} \vee \ldots \vee
  \phi_{[t_n]} \vee \phi_{[s_1]} \ldots \vee \phi_{[s_k]}$ is provably
  equivalent to $\top$. Consider $\vee_{\psi \in [t]} \phi_{[t]}$. It
  is provably equivalent to $(\psi \wedge \phi_{[t_1]}) \vee \ldots
  \vee (\psi \wedge \phi_{[t_n]})$.  Since for every $[s_i]$ such that
  $\neg \psi \in [s_i]$, $(\psi \wedge \phi_{[s_i]})$ is provably
  equivalent to $\bot$,
$$(\psi \wedge \phi_{[t_1]}) \vee \ldots \vee (\psi \wedge \phi_{[t_n]})$$
is provably equivalent to 
$$(\psi \wedge \phi_{[t_1]}) \vee \ldots \vee (\psi \wedge \phi_{[t_n]}) \vee (\psi \wedge \phi_{[s_1]}) \vee \ldots \vee (\psi \wedge \phi_{[s_k]})$$
which in turn is provably equivalent to
$$\psi \wedge (\phi_{[t_1]} \vee \ldots \vee \phi_{[s_k]})$$
which in turn is equivalent to $\psi \wedge \top$ hence to $\psi$. So in $M^c$, $\{s': (\vee_{\psi \in [t]} \phi_{[t]}) \in s'\} = \{s': \psi \in s'\}$. 

Proof of (**): since we defined $X \in E^c(\agents)(s)$ to hold iff
$S^c \setminus X \not\in E^c(\emptyset)(s)$, it suffices to show the
case that $G \neq \agents$.  The direction to the left is immediate:
if $[G]\psi \in s$ then $\{s' \in S^c : \psi \in s'\} \in E^c(G)(s)$
by definition. For the other direction assume that $\{s' \in S^c :
\psi \in s'\} \in E^c(G)(s)$, i.e., there is some $\gamma$ such that
$\{s' \in S^c : \gamma \in s'\} \subseteq \{s' \in S^c : \psi \in
s'\}$ and $[G]\gamma \in s$. It is easy to see that $\{s' \in S^c :
\gamma \in s'\} \subseteq \{s' \in S^c : \psi \in s'\}$ implies that
$\vdash \gamma \rightarrow \psi$, and by the monotonicity rule it
follows that $[G]\psi \in s$.

\item[case] $\theta=C_G \psi$ 
The proof is similar to \cite{hvdetal.del:2007}. First we show that in $M^f$, if $C_G \psi \in cl(\phi)$, then $C_G \psi \in [s]$
iff every state on every $\sup_{i \in G} \sim_i^f$ path from $[s]$ contains $\psi$.

Suppose $C_G \psi \in [s]$. The proof is by induction on the length of the path. If the path is of 0 length, then clearly by deductive closure
and by $\psi \in cl(\phi)$ we have $\psi \in [s]$. We also have $C_G \psi \in [s]$ by the assumption.
IH: if $C_G \psi \in [s]$, then every state on every $\cup_{i \in G} \sim_i^f$ path of length $n$ from $[s]$ contains $\psi$ \emph{and $C_G \psi$}.
Inductive step: let us prove this for paths of length $n+1$. Suppose we have a path $[s] \sim_{i_1}^f [s_1] \ldots \sim_{i_{n}}^f [s_n]
\sim_{i_{n+1}}^f [s_{n+1}]$. By the IH, $\psi, C_G \psi \in [s_n]$. Since $s_n$ is deductively closed and 
$K_{i_{n+1}} C_G \psi \in cl(\phi)$, we have $K_{i_{n+1}} C_G \psi \in  [s_n]$. Since $[s_n] \sim_{i_{n+1}}^f [s_{n+1}]$
and the definition of $\sim_{i_{n+1}}^f$, $C_G \psi \in  [s_{n+1}]$ and hence by reflexivity $\psi \in  [s_{n+1}]$.

For the other direction, suppose that every state on every $\cup_{i \in G} \sim^f$ path from $[s]$ contains $\psi$. Prove that
$C_G \psi \in [s]$. Let $S_{G,\psi}$ be the set of all $[t]$ such that  every state on every $\cup_{i \in G} \sim^f$ path from $[t]$ 
contains $\psi$. Note that each $[t]$ is/corresponds to a finite set of formulas so we can write its conjunction $\phi_{[t]}$.
Consider a formula 
$$\chi = \bigvee_{[t] \in S_{G,\psi}} \phi_{[t]}$$
Similarly to \cite{hvdetal.del:2007} it can be proved that $\vdash_{CLC}\phi_{[s]} \rightarrow \chi$, $\vdash_{CLC}\chi \rightarrow \psi$
and $\vdash_{CLC}\chi \rightarrow  E_G \chi$. And from that follows that $\vdash_{CLC}\phi_{[s]} \rightarrow C_G \psi$
hence $C_G \psi \in [s]$.

Now we prove that $M^f, [s] \models C_G \psi$ iff $C_G \psi \in [s]$. $C_G \psi \in [s]$ iff every
state on every $\cup_{i \in G} \sim_i^f$ path from $[s]$ contains $\psi$ iff 
for every $[t]$ reachable from $[s]$ by a $\cup_{i \in G} \sim_i^f$ path, $M^f,[t] \models \psi$ iff $M^f, [s] \models C_G \psi$. 
 
\end{description}
\end{proof}

It is obvious that in $M^f$, $\sim_i$ are equivalence relations. So
what remains to be proved is that $E^f$ satisfies E1-E6. Since $S^f$
is finite, it suffices to show E1-E5, which for finite sets of states
entail E6.

\begin{proposition}
$M^f$ satisfies E1-E5.
\end{proposition}
\begin{proof}
\begin{description}
\item[E1] Note that $\phi_{\emptyset}$ is an empty disjunction, namely $\bot$.

$\emptyset \in E^f(G)([s])$ iff (by definition of $E^f$)
$\{s': \bot \in s'\}   \in E^c(G)(s)$ iff $\emptyset \in E^c(G)(s)$. Since 
$E^c$ satisfies $E1$, $\emptyset \not \in E^f(G)([s])$.

\item[E2] $S^f \in E^f(G)([s])$ iff 
$\{s': \bigvee_{[t] \in S^f} \in s'\} \in E^c(G)(s)$ iff 
$S^c \in E^c(G)(s)$. Since $E^c$ satisfies $E2$,
$S^f \in  E^f(G)([s])$.

\item[E3] Let $\bar{X} \not \in E^f(\emptyset)([s])$. Then $\{s': \phi_{\bar{X}} \in s'\} \not \in E^c(\emptyset)(s)$. Note that
$\{s': \phi_{\bar{X}} \in s'\}$ is the complement of $\{s': \phi_{X} \in s'\}$, since $\phi_{\bar{X}} = \neg \phi_{X}$.
Since $E^c$ satisfies E3, this means that $\{s': \phi_{X} \in s'\} \in E^c(N)(s)$. Hence $X \in  E^f(N)([s])$.
 
\item[E4] Let $X \subseteq Y \subseteq S^f$ and $X \in E^f(G)([s])$. Clearly $\vdash_{CLC}  \phi_X \rightarrow \phi_Y$. 
Hence $\{s': \phi_X \in s'\} \subseteq \{s': \phi_Y \in s'\}$. Since  $X \in E^f(G)([s])$, we have $\{s': \phi_X \in s'\} \in E^c(G)(s)$.
Since $E^c$ satisfies E4,  $\{s': \phi_Y \in s'\} \in E^c(G)(s)$ so $Y \in E^f(G)([s])$.

\item[E5]  Let $X \in E^f(G_1)([s])$ and $Y \in E^f(G_2)([s])$ and $G_1 \cap G_2 = \emptyset$. 
So $\{s': \phi_X \in s'\} \in E^c(G_1)(s)$ and $\{s': \phi_Y \in s'\} \in E^c(G_2)(s)$ and since $E^c$ satisfies E5,
$\{s': \phi_X \in s'\} \cap \{s': \phi_Y \in s'\} \in  E^c(G_2)(s)$. Note that 
$$\{s': \phi_X \in s'\} \cap \{s': \phi_Y \in s'\} = \{s': (\vee_{[t] \in X} \phi_{[t]}) \in s'\ {\rm and}\ (\vee_{[t] \in Y} \phi_{[t]}) \in s'\}$$
which is in turn the same as
$$\{s': (\vee_{[t]  \in X \cap Y} \phi_{[t]} \in s'\}$$
since $\{s': (\vee_{[t]  \in X \cap Y} \phi_{[t]} \in s'\} \in  E^c(G_2)(s)$, $X \cap Y \in E^f(G_1)([s])$. 
\end{description}
\end{proof}

\begin{corollary}
  For any $CLC$-formula $\phi$,  $\vdash_{CLC} \phi$ iff $\models \phi$.
\end{corollary}

\section{Completeness of Coalition Logic with Distributed Knowledge}
\label{sec:cld}

In this section we consider the logic \acro{CLD}, extending coalition
logic with individual knowledge operators and distributed knowledge.

The axiomatisation $CLD$ is shown in Figure \ref{fig:cld-ax}. It extends
$CL$ with standard axioms and rules for individual and distributed
knowledge (see, e.g., \cite{Fagin1995}).

\begin{figure}[h]
  \centering
\begin{description}
\item[Prop] Classical propositional logic
\item[G1] $\neg [G] \bot$
\item[G2] $[G] \top$
\item[G3] $\neg [\emptyset] \neg \phi \rightarrow [N] \phi$
\item[G4] $[G] (\phi \wedge \psi) \rightarrow [G] \psi$
\item[G5] $[G_1]\phi \wedge [G_2]\psi \rightarrow [G_1 \cup G_2] (\phi \wedge \psi)$, if $G_1 \cap G_2 = \emptyset$
\item[MP] $\vdash_{CLD} \phi, \phi \rightarrow \psi \Rightarrow \vdash_{CLD}  \psi$
\item[RG] $\vdash_{CLD} \phi \leftrightarrow \psi \Rightarrow \vdash_{CLD} [G] \phi \leftrightarrow [G] \psi $
\item[K] $K_i (\phi \rightarrow \psi) \rightarrow (K_i \phi \rightarrow K_i \psi)$
\item[T] $K_i \phi \rightarrow \phi$
\item[4] $K_i \phi \rightarrow K_i K_i \phi$
\item[5] $\neg K_i \phi \rightarrow K_i \neg K_i \phi$
\item[RN] $\vdash_{CLD} \phi \Rightarrow \vdash_{CLD} K_i \phi$
\item[DK] $D_G (\phi \rightarrow \psi) \rightarrow (D_G \phi \rightarrow D_G \psi)$
\item[DT] $D_G \phi \rightarrow \phi$
\item[D4] $D_G \phi \rightarrow D_G D_G \phi$
\item[D5] $\neg D_G \phi \rightarrow D_G \neg D_G \phi$
\item[D1] $K_i \phi \leftrightarrow D_i \phi$
\item[D2] $D_G \phi \rightarrow D_H\phi$, if $G \subseteq H$
\end{description}
  \caption{$CLD$: axiomatisation of \acro{CLD}.}
  \label{fig:cld-ax}
\end{figure}

As usual, soundness can easily be shown.
\begin{lemma}[Soundness]
  For any $CLD$-formula $\phi$, $\vdash_{CLD} \phi \Rightarrow  \models \phi$.
\end{lemma}
In the remainder of this section we show that $CLD$ also is complete.

For a set of formulae $s$, let $K_a s= \{K_a \phi : K_a
\phi \in s\}$ and $D_G s = \{D_G \phi : D_G \phi \in
s\}$.

\begin{definition}[Canonical Playable Pseudomodel]
  The canonical playable pseudomodel $M^c = (S^c,\{\sim_i^c : i \in \agents\},
  \{R_G^c : \emptyset \neq G \subseteq \agents\}, E^C, V^c)$ for \acro{CLD} is defined as
  follows:
  \begin{itemize}
  \item $S^c$ is the set of maximal consistent sets.
  \item $s \sim_i^c t$ iff $K_a s= K_a t$
  \item $s R_G t $ iff $D_H s = D_H t$ whenever $H \subseteq G$
  \item $V^c(p) = \{s \in S^c : p \in s\}$
  \item $X \in E^c(G)(s)$ iff $\left\{\begin{array}{ll}
        \exists \phi \{s \in S^c : \phi \in s\} \subseteq X : [G]\phi \in s& G \neq N\\
        \forall \phi \{s \in S^c : \phi \in s\} \subseteq S^c\setminus X : [\emptyset]\phi \not\in s& G = N
        \end{array}\right.$
  \end{itemize}
\end{definition}

\begin{lemma}[Pseudo Truth Lemma]
  $M^c,s \models \phi \Leftrightarrow \phi \in s$.
\end{lemma}
\begin{proof}
  The proof is by induction on $\phi$. The epistemic cases are exactly
  as for standard normal modal logic. The case for coalition operators
  is exactly as in \cite{pauly:2002a}.
\end{proof}

It is easy to check that $\sim^c_i$ are equivalence relations and E1-E5 hold
for $E^c$.

\begin{lemma}[Finite Pseudomodel]\label{lemma:cld}
Every $CLD$-consistent formula $\phi$ has a finite pseudomodel where E1-E6
hold.
\end{lemma}
\begin{proof}
The proof is exactly as in Theorem~\ref{th:clc}, namely the
construction of $M^f$, but starting with a Canonical Playable Pseudomodel
rather than Canonical Playable Model; the definition of $M^c$
contains the clause
 \[\Gamma R_G \Delta \mbox{ iff } \forall H \subseteq G \{\psi : D_H\psi \in \Gamma\} = \{\psi: D_H\psi \in \Delta\}\]

We add the following condition to the closure: $D_i\phi \in cl(\phi)$ iff $K_i \psi \in cl(\phi)$.

We define $M^f$ to be a pseudomodel instead of a model, by adding the clause:
    \[[s] R_G^f [s'] iff \forall H \subseteq G \{\psi : D_H \psi \in [s]\} = \{\psi : D_H\psi \in [s']\}\]

We show that $M^f$ is indeed a pseudomodel:
    \begin{itemize}
    \item $R^f_i = \sim_i^f$: this follows from the fact that $K_i
      \phi \in [s]$ iff $D_i\phi \in [s]$ for any $i, \phi$ and $s$,
      which holds because of the $K_i \phi \rightarrow D_i\phi$ axiom
      and the new closure condition above.
   \item $G \subseteq H \Rightarrow R^f_H \subseteq R^f_G$: this holds
      by definition.
\end{itemize}

We add a case for $\theta = D_G\psi$ to the inductive
proof. This case is proven in exactly the same way as the $\theta
= K_i\psi$ case: the definitions of $\sim_i^f$ and $R_G^f$ are of
exactly the same form (in particular, $R_G^f$ is also an $S5$
modality). The proof that E1-E6 hold in the resulting pseudomodel is
the same as in the proof of Theorem~\ref{th:clc} for $E^f$.
\end{proof}

We are now going to transform the pseudomodel into a proper
model; it is a well-known technique for dealing with distributed
knowledge. In fact, we can make direct use of a corresponding existing
result for epistemic logic with distributed knowledge, and extend it
with the coalition operators/effectivity functions. We here give the
more general result for the language with also common knowledge, which
will be useful later.

\begin{theorem}[\cite{Fagin1995}]
  \label{th:yi}
  If $M = (S, \{\sim_i : i \in \agents\}, \{R_G : \emptyset \neq G
  \subseteq \agents\}, V)$ is an epistemic pseudomodel, then there is
  an epistemic model $M' = (S',\{\sim_i' : i \in \agents\}, V')$ and a
  surjective (onto) function $\themap: S' \rightarrow S$ such that
 for every $s' \in S'$ and formula $\phi \in \acro{ELCD}$, $M,\themap(s')
    \models \phi$ iff $M',s' \models \phi$.
\end{theorem}
\begin{proof}
  This result is directly obtained from the completeness proof for
  \acro{ELCD} sketched in \cite[p. 70]{Fagin1995}. For a more detailed proof
  (for a more general language), see \cite[Theorem 9]{yi}.
\end{proof}

\begin{theorem}
  \label{th:sat-d}
  If a formula is satisfied in a finite pseudomodel, then it is satisfied 
in a model.
\end{theorem}
\begin{proof}
  Let $M = (S, \{\sim_i : i \in \agents\}, \{R_G : \emptyset \neq G
  \subseteq \agents\}, E, V)$ be a finite pseudomodel such that $M,s \models
  \phi$.  Let $M_p = (S, \{\sim_i : i \in \agents\}, \{R_G : \emptyset
  \neq G \subseteq \agents\}, V)$ be the epistemic pseudomodel
  underlying $M$, and let $M'_p = (S',\{\sim_i' : i \in \agents\},
  V')$ and $\themap: S' \rightarrow S$ be as in Theorem \ref{th:yi}.  Let
  $\img(X) = \{s' : \themap(s') \in X\}$ for any set $X \subseteq
  S$. Finally, let $M' = (S',\{\sim_i' : i \in \agents\}, E', V')$
  where $E'$ is defined as follows:
  \begin{itemize}
  \item For $G \not = N$:
$$Y \in E'(G)(u) \ \Leftrightarrow\ \exists X \subseteq S, \ (Y \supseteq \img(X) {\rm and}\ X \in E(G)(\themap(u)))$$
\item for $G = N$:
$$Y \in E'(G)(u) \ \Leftrightarrow\ \bar{X} \not \in E'(\emptyset)(u)$$
  \end{itemize}

  Two things must be shown: that $M'$ is a proper model, and that it satisfies $\phi$.

  Since $M'_p$ is an epistemic model, to show that $M'$ is a model all
  that remains to be shown is that $E'$ is truly playable. We now show
  that that follows from true playability of $E$.
  \begin{description}
\item[E1] Note that $\img(X)=\emptyset$ iff $X = \emptyset$. 

For $G \not = N$,
$\emptyset \in E'(G)(u)$ iff (by definition of $E'$)
$\exists X \subseteq S, s \in S \ (\emptyset \supseteq \img(X) {\rm and}\ X \in E(G)(\themap(u)))$ iff 
$\emptyset \in E(G)(\themap(u)))$ which is impossible since $M$ satisfies E1.
Note that in particular this proves $\emptyset \not \in E'(\emptyset)(u)$. 

For $N$, $\emptyset \in E'(N)(u)$ iff $S' \not \in E'(\emptyset)(u)$ and we'll see that this is impossible below.

\item[E2] Note that $\img(S)=S'$.

For $G \not = N$, 
$S' \in E'(G)(u)$ iff (by definition of $E'$) $\exists X \subseteq S, \ (S' \supseteq \img(X) {\rm and}\ X \in E(G)(\themap(u)))$ 
and since $S' \supseteq \img(S)$ and $S \in E(G)(\themap(u))$, $S' \in E'(G)(u)$ holds. 
Note that in particular this proves $S' \in E'(\emptyset)(u)$.

For $N$, $S' \in E'(N)(u)$ iff $\emptyset \not \in E'(\emptyset)(u)$ and this was proved above.

\item[E3] $\forall u \in S' \forall Y \subseteq S' \ \bar{Y} \not \in E'(\emptyset)(u) \Rightarrow Y \in E'(N)(u)$ follows immediately
from the definition for $E'(N)$.

\item[E4] $E'$ is monotonic by definition for $G \not = N$. 

For $N$, assume $X \subseteq Y$ and $X \in E'(N)(u)$. Then $\bar{X} \not \in E'(\emptyset)(u)$. Since for $\emptyset$ we already have
monotonicity and $\bar{Y} \subseteq \bar{X}$, $\bar{Y} \not \in E'(\emptyset)(u)$. So $Y \in E'(N)(u)$.

\item[E5] $\forall u \in S' \forall G_1,G_2 \subseteq N \forall X',Y' \subseteq S'\ X' \in E'(G_1)(u)$ and $Y' \in E'(G_2)(u)$
$\Rightarrow X' \cap Y' \in E'(G_1 \cup G_2)(u)$, where $G_1 \cap G_2 = \emptyset$

For $G_1,G_2 \not = N$:
 
Let $s = \themap(u)$, $G_1  \cap G_2 = \emptyset$, $X' \in E'(G_1)(u)$, $Y' \in E'(G_2)(u)$. This means that for some $X$, $Y$,
such that $X' \supseteq X$ and $Y' \supseteq Y$, 
$X  \in E (G_1)(s)$ and $Y  \in E (G_2)(s)$, so since $M$ satisfies E5, $X  \cap Y  \in E(G_1 \cup G_2)(s)$. Since
$X' \supseteq X$ and $Y' \supseteq Y$, $X' \cap Y' \supseteq X \cap Y$, so $ X' \cap Y' \in E'(G_1 \cup G_2)(u)$.

For $G_1=N$ ($G_2$ has to be $\emptyset$): $X' \in E'(N)(u)$, $Y' \in E'(\emptyset)(u)$, prove $X' \cap Y' \in E'(N)(u)$.
We have $\bar{X'} \not \in E'(\emptyset)(u)$, $Y' \in E'(\emptyset)(u)$.  Assume by contradiction $X' \cap Y' \not \in E'(N)(u)$,
then $\bar{X' \cap Y'} \in E'(\emptyset)(u)$ which means $\bar{X'} \cup \bar{Y'} \in E'(\emptyset)(u)$. This together with
$Y' \in E'(\emptyset)(u)$ and E5 for $G_1,G_2=\emptyset$ gives $(\bar{X'} \cup \bar{Y'})\cap Y' \in E'(\emptyset)(u)$, that is
$\bar{X'} \in E'(\emptyset)(u)$. The latter contradicts $\bar{X'} \not \in E'(\emptyset)(u)$.

\item[E6] Suppose $X \in E^{nc}(\emptyset)(s)$. We claim that for
  every $u$ such that $s = \themap(u)$, $\img(X) \in
  E^{'nc}(\emptyset)(u)$.

  Assume by contradiction that there exists a $Y \subset \img(X)$ such
  that $Y \in E^{'nc}(\emptyset)(u)$ for some $u$ s.t. $s =
  \themap(u)$. By the definition of $E'$, this means that there is a
  $Z$ such that $\img(Z) \subseteq Y$ and $Z \in
  E(\emptyset)(s)$. Since $Y \subset \img(X)$ and $\img(Z) \subseteq
  Y$ it follows that $\img(Z) \subset \img(X)$. But it is a property
  of the $\img$ function that that implies that $Z \subset Y$, and
  this contradicts the assumption that $Z \in E(\emptyset)(s)$ and $X
  \in E^{nc}(\emptyset)(s)$.

\end{description}

In order to show that $M'$ satisfies $\phi$, we show that $M,\themap(u)
\models \gamma$ iff $M',u \models \gamma$ for $u \in S'$ and any
$\gamma$, by induction in $\gamma$. All cases except $\gamma =
[G]\phi$ are exactly as in the proof of Theorem \ref{th:yi} (see
\cite[Theorem 9]{yi} for a detailed inductive proof).

For the case that $\gamma = [G]\phi$, the inductive hypothesis is that
for all $\psi$ with $|\psi| < |[G] \phi|$, and any $t,v$ with $t =
\themap(v)$, $M,t \models \psi$ iff $M',v \models \psi$.  Given that
for every $v$ there is a unique $t$ such that $t = \themap(v)$, we can
state this as $\{v: M', v \models \psi\} = \img(\psi^M)$, or
$\psi^{M'} = \img(\psi^M)$.

First consider $G \not = N$.  $M,s \models [G] \phi$ iff $\phi^M \in
E(G)(s)$. Consider $\img(\phi^M)$. By the inductive hypothesis,
$\phi^{M'} = \img(\phi^M)$.  $\phi^M \in E(G)(s)$ holds iff $\phi^{M'}
\in E'(G)(u)$ iff $M',u \models [G] \phi$.

$M,s \models [N] \phi$ iff $\phi^M \in E(N)(s)$ iff (*) $\neg \phi^M
\not \in E(\emptyset)(s)$ iff (as above) $\neg \phi^{M'} \not \in
E'(\emptyset)(u)$ iff $M',u \models [N] \phi$.

Explanation of (*): one direction E3, the other direction E5 and E1.
\end{proof}

\begin{corollary}
  For any $CLD$-formula $\phi$,  $\vdash_{CLD} \phi$ iff $\models \phi$.
\end{corollary}

\section{Completeness of Coalition Logic with both Common and Distributed Knowledge }
\label{sec:clcd}

In this section we consider the logic \acro{CLCD}, extending coalition
logic with operators for individual knowledge, common knowledge and
distributed knowledge.

The axiomatisation $CLCD$ is shown in Figure \ref{fig:clcd-ax}. It
extends $CL$ with standard axioms and rules for individual, common and
distributed knowledge.

\begin{figure}[h]
  \centering
\begin{description}
\item[Prop] Classical propositional logic
\item[G1] $\neg [G] \bot$
\item[G2] $[G] \top$
\item[G3] $\neg [\emptyset] \neg \phi \rightarrow [N] \phi$
\item[G4] $[G] (\phi \wedge \psi) \rightarrow [G] \psi$
\item[G5] $[G_1]\phi \wedge [G_2]\psi \rightarrow [G_1 \cup G_2] (\phi \wedge \psi)$, if $G_1 \cap G_2 = \emptyset$
\item[MP] $\vdash_{CLCD} \phi, \phi \rightarrow \psi \Rightarrow \vdash_{CLCD}  \psi$
\item[RG] $\vdash_{CLCD} \phi \leftrightarrow \psi \Rightarrow \vdash_{CLCD} [G] \phi \leftrightarrow [G] \psi $
\item[K] $K_i (\phi \rightarrow \psi) \rightarrow (K_i \phi \rightarrow K_i \psi)$
\item[T] $K_i \phi \rightarrow \phi$
\item[4] $K_i \phi \rightarrow K_i K_i \phi$
\item[5] $\neg K_i \phi \rightarrow K_i \neg K_i \phi$
\item[RN] $\vdash_{CLCD} \phi \Rightarrow \vdash_{CLCD} K_i \phi$
\item[C1] $E_G \phi \leftrightarrow  \bigwedge_{i \in G} K_i \phi$
\item[C2] $C_G  \phi \rightarrow E_G(\phi \wedge C_G \phi)$
\item[RN] $\vdash_{CLCD} \phi \Rightarrow \vdash_{CLCD} K_i \phi$
\item[RC] $\vdash_{CLCD} \phi \rightarrow E_G(\phi \wedge \psi) \Rightarrow \vdash_{CLCD} \phi \rightarrow C_G \psi$
\item[DK] $D_G (\phi \rightarrow \psi) \rightarrow (D_G \phi \rightarrow D_G \psi)$
\item[DT] $D_G \phi \rightarrow \phi$
\item[D4] $D_G \phi \rightarrow D_G D_G \phi$
\item[D5] $\neg D_G \phi \rightarrow D_G \neg D_G \phi$
\item[D1] $K_i \phi \leftrightarrow D_i \phi$
\item[D2] $D_G \phi \rightarrow D_H\phi$, if $G \subseteq H$
\end{description}
  \caption{$CLCD$: axiomatisation of \acro{CLCD}.}
  \label{fig:clcd-ax}
\end{figure}

As usual, soundness can easily be shown.
\begin{lemma}[Soundness]
  For any $CLCD$-formula $\phi$, $\vdash_{CLCD} \phi \Rightarrow  \models \phi$.
\end{lemma}
In the remainder of this section we show that $CLCD$ also is complete.

\begin{theorem}
  Any $CLCD$-consistent formula is satisfied in a finite pseudomodel.
\end{theorem}
\begin{proof}
  The proof is identical to the proof of Lemma \ref{lemma:cld}, with the
addition of the inductive clause $\theta= C_G \psi$ as in the proof of
Theorem \ref{th:clc}.
\end{proof}

We can now use the same approach as in the case of \acro{CLD}.

\begin{theorem}
  If a \acro{CLCD} formula is satisfied in a finite pseudomodel, it is
  satisfied in a model.
\end{theorem}
\begin{proof}
  The proof goes exactly like the proof of Theorem \ref{th:sat-d},
  using Theorem \ref{th:yi}.  The definition of the model $M'$ is
  identical to the definition in Theorem \ref{th:sat-d}, as is the proof 
that it is a
  proper model. For the last part of the proof, i.e., showing that
  $M'$ satisfies $\phi$, note that the last clause in Theorem
  \ref{th:yi} holds for epistemic logic with both distributed and
  common knowledge. Thus, the proof is completed by only adding the
  inductive clause for $[G]\phi$, which is done in exactly the same
  way as in Theorem \ref{th:sat-d}.
\end{proof}

\begin{corollary}
  For any $CLCD$-formula $\phi$,  $\vdash_{CLCD} \phi$ iff $\models \phi$.
\end{corollary}

\section{Conclusions}
\label{sec:conclusions}

This papers solves several hitherto open problems, namely proving
completeness of Coalition Logic extended with group knowledge
modalities. The axioms for the epistemic modalities are the same as in
the absence of the Coalition Logic axioms, however the completeness
proofs require non-trivial combinations of techniques. The next step
would be to look at complete axiomatisations of logics resulting from
imposing some conditions on the interaction of coalitional ability and
group knowledge (such as the examples in the Introduction), and
obtaining results on the complexity of satisfiability problem for
\acro{CLC}, \acro{CLD} and \acro{CLCD}.

\bibliographystyle{plain}
\bibliography{bib}

\end{document}